\documentclass[12pt]{article}
\usepackage[paperwidth=8.5in,paperheight=11in,
            left=1.25in,top=1in,
            textwidth=6in,textheight=9in]{geometry}
\usepackage{amsmath,amssymb,amsthm}
\usepackage{tikz}
\usetikzlibrary{positioning,arrows.meta}
\usepackage{booktabs}
\usepackage{multirow}
\usepackage[colorlinks=true,linkcolor=blue,citecolor=blue,urlcolor=blue]{hyperref}

\newtheorem{theorem}{Theorem}
\newtheorem{lemma}{Lemma}
\newtheorem{corollary}{Corollary}

\newcommand{\E}{\mathbb{E}}
\newcommand{\eff}{\mathrm{eff}}

\title{\bfseries Optimal Weighting of Training Data in Adaptive Coding}
\author{Yuriy A. Reznik\\[2pt]
\normalsize Massachusetts Institute of Technology, Cambridge, MA, USA\\
\normalsize \texttt{yreznik@mit.edu}}
\date{}

\begin{document}
\maketitle

\begin{abstract}
Adaptive coders are commonly primed with training data, turned into
context tables or shipped dictionaries. The training data generally
do not match the source being encoded, and the coder must decide how
much to trust them. We make
that trust a design variable: a Krichevsky--Trofimov (KT) estimator
over an~$m$-ary alphabet whose training counts are scaled
by~$\xi\in[0,1]$. For a training sequence of length~$\ell$ at
Kullback--Leibler divergence~$D$ nats per symbol from the message
source, the redundancy-minimizing weight is $\xi^*=d/(2\ell D+d)$,
$d=m-1$. The effective training length $\xi^*\ell$ follows a harmonic
law: any mismatch caps the usable training information at~$d/(2D)$
symbols, however much was collected.
Three implementable selectors supply the unknown~$D$: offline, from
the spread of the training set; by a plug-in loop, from the decoded
prefix; and by a twice-universal mixture, with no estimation at all.
On the ten English texts of the Calgary and Canterbury corpora,
weighting removes up to~$22.7\%$ of the redundancy of the classical
KT code, the $\xi=0$ endpoint, and beats both endpoints on every
file. An open-source implementation is provided.
\end{abstract}

\section{Introduction}

Adaptive coders are rarely deployed cold: a training dataset,
assembled offline, primes the probability estimates that the coder
then updates as the message arrives. Video coders initialize their
context models from tables trained on test corpora~\cite{Marpe2003}.
Dictionary coders ship trained dictionaries to jump-start the
compression of small records~\cite{Zstd}. Network packet compressors
memorize previously transmitted sequences as common side information~\cite{BeiramiFekri2012}. Pretrained models supply coding distributions
directly~\cite{Deletang2024}. In every case the same design question
arises: the training data need not come from the source being encoded,
and the coder must decide how much to trust them.

At its extremes the question has classical answers. With no training
data, a universal code for the memoryless class over an~$m$-ary
alphabet pays redundancy $\frac{d}{2}\ln n + O(1)$ nats on
messages of length~$n$, $d=m-1$~\cite{KT1981,Rissanen1984,ClarkeBarron1990}. The KT estimator attains
it~\cite{KT1981}. With training data from the \emph{same} source, a matched sample
of length~$\ell$ in the KT counts reduces this to
$\frac{d}{2}\ln\frac{\ell+n}{\ell}+O(1)$ nats~\cite{KT1981,ReznikSzpankowski2003}, the memory-assisted optimum of~\cite{BeiramiFekri2012}. Between the extremes stands the mismatch, and
the treatments to date control the training data by \emph{selection}:
under a per-symbol divergence~$D$, the redundancy-minimizing amount of
full-weight training data is finite, about~$d/(2D)$ symbols, and
more \emph{increases} the redundancy~\cite{ReznikAnisimov2003a,ReznikAnisimov2003b}. For mixture sources
the first-order optimal use of a heterogeneous memory is clustering,
full-weight use of the matched subset~\cite{BeiramiSardariFekri2013}.
Practice follows the same logic: one dictionary per data type, context
tables per coding mode. The control variables have been \emph{which}
training data and \emph{how much}. The weight at which the retained
data enter the code has not been a variable at all.

This paper makes it the variable. The construction is imported from
Bayesian statistics: the \emph{power prior} of Ibrahim and Chen~\cite{IbrahimChen2000} incorporates a previously observed dataset by
raising its likelihood to a power~$a_0\in[0,1]$, conceived from the
outset as a scalar expressing the difference between the old data and
the new. A related draft develops the theory of optimal discounting in
this broader statistical context~\cite{ReznikDiscounting2026}. Applied to the KT estimator, the construction is one line,
scale the training counts by~$\xi\in[0,1]$, and in the coding setting
the choice of the exponent has a closed-form answer. For a training
sequence of length~$\ell$ at Kullback--Leibler divergence~$D$ nats per
symbol from the message source, the redundancy-minimizing weight is
\begin{equation}\label{eq:xistar-intro}
\xi^* \;=\; \frac{d}{2\ell D + d},
\qquad d = m-1,
\end{equation}
and the optimal effective training length $\ell^*_\eff=\xi^*\ell$
obeys the \emph{harmonic combination law} $1/\ell^*_\eff = 1/\ell + 2D/d$.
At $D=0$ the training data are pooled in full, $\xi^*=1$. Any
mismatch caps the usable training information at
$\ell_\infty=d/(2D)$ symbols, however much was collected. Weighting strictly beats truncating
(Section~\ref{sec:law}).

The practical content of the paper is what to do when~$D$ is unknown,
which it always is. Section~\ref{sec:algs} develops three
implementable selectors, sketched in Figs.~\ref{fig:offline}--\ref{fig:mixture}: an
\emph{offline} weight from the spread of the training set, a
\emph{plug-in} loop that re-estimates~$D$ from the decoded prefix, and
a twice-universal \emph{mixture} over a grid of weights that needs
no estimate at all.
Section~\ref{sec:exp} evaluates all three on the ten English texts of
the Calgary and Canterbury corpora. Section~\ref{sec:concl} brings
conclusions and discusses applications.

\section{The Weighted KT Code}
\label{sec:setting}

Symbols take values in the alphabet $\mathcal{X}=\{1,\dots,m\}$. The
training sequence~$x_T$ consists of~$\ell$ independent draws from a
distribution~$P_T$. The message $x_S=(x_{S,1},\dots,x_{S,n})$ consists
of~$n$ independent draws from a distribution~$P_S$, independent of the
training. Both distributions are unknown to the coder and need not
coincide. Their mismatch is the \emph{Kullback--Leibler divergence}
$D=D(P_T\,\|\,P_S)=\sum_{x} P_T(x)\ln(P_T(x)/P_S(x))$. Information is
measured in nats throughout the analysis, with~$\ln$ the natural logarithm. Section~\ref{sec:exp} converts to bits
for reporting. Throughout, $d=m-1$ is the model dimension: $d/2$ nats
is the expected amount by which a fitted~$m$-ary model over-describes
its own sample, the \emph{Wilks excess}~\cite{Wilks1938}, and it is the
constant that recurs in many formulae below. For a sequence~$u$,
$k_x(u)$ counts the occurrences of symbol~$x$ in~$u$.

The coder is a KT estimator~\cite{KT1981} whose training counts enter
once, scaled by a weight~$\xi\in[0,1]$, and whose message counts enter
at full weight:
\begin{equation}\label{eq:kt}
Q_\xi(x_{S,t}=x \mid x_T,\, x_{S,1}\cdots x_{S,t-1})
 = \frac{\xi\, k_x(x_T) + k_x(x_{S,1}\cdots x_{S,t-1}) + \tfrac12}
        {\xi\ell + t-1 + \tfrac{m}{2}} .
\end{equation}
An arithmetic coder driven by these conditionals produces a code of
length~$-\ln Q_\xi$ nats, up to rounding. Construction \eqref{eq:kt}
is the power prior of Ibrahim and Chen~\cite{IbrahimChen2000}
instantiated for the multinomial family with the Jeffreys initial
prior: the trained prior is the Dirichlet distribution with parameters
$\xi k_x(x_T)+\frac12$, and~$\xi$ plays the role of the discounting
parameter~$a_0$. At~$\xi=1$ the
training data are treated as part of the message. At~$\xi=0$ they are
discarded and~\eqref{eq:kt} is the classical KT estimator. Since
scaling the counts of~$\ell$ symbols by~$\xi$ scales their information
content by~$\xi$, the natural gauge of the trained model's strength is
the \emph{effective training length} $\ell_\eff=\xi\ell$: the number of
matched symbols that would be as informative.

Performance is measured by the \emph{average redundancy}
$R_n(\xi)=\E_{x_T}\,D(P_S^{\,n}\,\|\,Q_\xi(\cdot\mid x_T))$ nats: the
expected excess code length over the entropy of the message, averaged
over the training sequence, with $Q_\xi(\cdot\mid x_T)$ the joint
distribution assigned by~\eqref{eq:kt}. By the chain rule it is the
sum of the per-symbol excesses. Its values at the endpoints were
quoted in the introduction. The question is what happens between them
when the sources differ.

\section{Redundancy Analysis}
\label{sec:law}

The analysis needs one classical object. Write
$H(p)=-\sum_x p(x)\ln p(x)$ for the entropy in nats and
$M_\gamma=\gamma P_T+(1-\gamma)P_S$ for the mixture of the two
sources. The \emph{skewed Jensen--Shannon divergence}~\cite{Lin1991} between~$P_T$ and~$P_S$ at skew~$\gamma$ is
\begin{equation}\label{eq:JS}
\mathrm{JS}_\gamma(P_T,P_S)
 \;=\; H(M_\gamma)-\gamma H(P_T)-(1-\gamma)H(P_S).
\end{equation}
It measures the information
that one symbol carries about which of the two sources produced it,
when the two are mixed in proportions~$\gamma$ and~$1-\gamma$. Three
of its properties are needed.

\begin{lemma}\label{lem:JS}
For interior~$P_T,P_S$ and~$\lambda>0$, with
$\gamma_s=\lambda/(\lambda+s)$:
\begin{align}
\frac{d}{ds}\Bigl[(\lambda+s)\,\mathrm{JS}_{\gamma_s}(P_T,P_S)\Bigr]
 &= D(P_S\,\|\,M_{\gamma_s}), \label{eq:JSderiv}\\
\lim_{\gamma\to 0}\ \mathrm{JS}_\gamma(P_T,P_S)/\gamma
 &= D(P_T\,\|\,P_S) = D, \label{eq:JSlimit}\\
D(P_S\,\|\,M_\gamma) = \gamma^2 D_2 &+ O(\gamma^3),
 \label{eq:JSquad}
\end{align}
with $D_2=\frac{1}{2}\sum_x(P_T(x)-P_S(x))^2/P_S(x)$ and
$D_2=D\,(1+O(\|P_T-P_S\|))$.
\end{lemma}

The proofs are one-line computations: differentiate
$\nu H(M_\gamma)$ in~$\nu=\lambda+s$ for \eqref{eq:JSderiv}, split
$\mathrm{JS}_\gamma=\gamma D(P_T\|M_\gamma)+(1-\gamma)D(P_S\|M_\gamma)$
for \eqref{eq:JSlimit}, and Taylor-expand about~$M_\gamma=P_S$ for
\eqref{eq:JSquad}.

With these identities, the redundancy of the weighted code admits an
expansion valid at every ratio of training to message information.

\begin{theorem}\label{thm:exact}
Let $\ell_\eff=\xi\ell$ and $\gamma=\ell_\eff/(\ell_\eff+n)$. Then, as
$\ell_\eff\to\infty$ with~$\xi\in(0,1]$ fixed or vanishing, in nats,
\begin{equation}\label{eq:exact}
R_n(\xi)
 = \frac{d}{2}\ln\frac{\ell_\eff+n}{\ell_\eff}
 + (\ell_\eff+n)\,\mathrm{JS}_\gamma(P_T,P_S)
 - \frac{(1-\xi)\,d\,n}{2(\ell_\eff+n)}
 + O\bigl(\ell_\eff^{-1/2}\bigr) + O(1)^*,
\end{equation}
where~$O(1)^*$ collects terms bounded uniformly in~$(\xi,\ell,n)$.
\end{theorem}

\begin{proof}[Proof sketch]
By the chain rule the redundancy is the sum $R_n(\xi)=\sum_t\rho_t$
of the per-symbol excess code lengths
\begin{equation}\label{eq:rhodef}
\rho_t \;=\; \E\,D\bigl(P_S\,\big\|\,
  Q_\xi(\cdot\mid x_T,\,x_{S,1}\cdots x_{S,t-1})\bigr).
\end{equation}
Write~$s=t-1$. The predictive probabilities \eqref{eq:kt} concentrate
about the composition $M_{\gamma_s}$,
$\gamma_s=\ell_\eff/(\ell_\eff+s)$, with coordinate variances
$[\xi^2\ell P_T(1{-}P_T)+s P_S(1{-}P_S)]/(\ell_\eff+s)^2$ by
independence of the two count sets. In the second-order expansion of~$\rho_t$ about~$M_{\gamma_s}$, the bias of the~$\frac12$ pseudo-counts
cancels the source-dependent part of the variance term, and, with
$\xi^2\ell=\xi\ell_\eff$, the expansion collapses to
\begin{equation}\label{eq:perstep}
\rho_t = D(P_S\,\|\,M_{\gamma_s})
 + \frac{d}{2}\cdot\frac{\xi\ell_\eff+s}{(\ell_\eff+s)^2}
 + O\Bigl(\tfrac{\|P_T-P_S\|}{\ell_\eff+s}\Bigr)
 + O\Bigl(\tfrac{1}{(\ell_\eff+s)^{3/2}}\Bigr).
\end{equation}
The cancellation is what makes a single weight formula possible: the
estimation cost carries the universal coefficient~$d/2$ whether the
trained model is centered well or badly, and all of the mismatch sits
in the divergence term. Summing over~$t$ with sums replaced by
integrals at~$O(1)$ total error, \eqref{eq:JSderiv} with
$\lambda=\ell_\eff$ integrates the first term to the Jensen--Shannon
term of~\eqref{eq:exact}, and the second term integrates to the
remaining two.
\end{proof}

The regime of practice is two-scale: the effective training
information is ample, $\xi\ell\gg 1$, yet small against the message,
$\xi\ell\ll n$. Four corollaries follow: the redundancy in this
regime, the weight that minimizes it, the harmonic form of the
optimum, and the value attained.

\begin{corollary}[Two-scale redundancy]\label{cor:twoscale}
In the regime $1\ll\xi\ell\ll n$,
\begin{equation}\label{eq:main}
R_n(\xi) = \frac{d}{2}\ln\frac{n}{\xi\ell}
 + \xi\Bigl(\ell D+\frac{d}{2}\Bigr) -\frac{d}{2}
 + o(1) + \epsilon,
\end{equation}
where~$\epsilon$ collects the second-order remainders in the mismatch
and the~$O(\xi\ell/n)$ terms of the regime.
\end{corollary}

\begin{proof}
Apply \eqref{eq:JSlimit} to the middle term of~\eqref{eq:exact} and
expand the remaining terms in~$\xi\ell/n$.
\end{proof}

\begin{corollary}[Optimal weight]\label{cor:weight}
For~$D>0$, the redundancy~\eqref{eq:main} is minimized at
\begin{equation}\label{eq:xistar}
\xi^* \;=\; \frac{d}{2\ell D + d}.
\end{equation}
\end{corollary}

\begin{proof}
The~$\xi$-dependent part of~\eqref{eq:main},
$-\frac{d}{2}\ln\xi+\xi(\ell D+\frac{d}{2})$, is strictly convex with
stationary point \eqref{eq:xistar}.
\end{proof}

\begin{corollary}[Harmonic law]\label{cor:law}
The optimal effective training length $\ell^*_\eff=\xi^*\ell$
satisfies
\begin{equation}\label{eq:harmonic}
\frac{1}{\ell^*_\eff} \;=\; \frac{1}{\ell}+\frac{2D}{d}
 \;=\; \frac{1}{\ell}+\frac{1}{\ell_\infty},
 \qquad \ell_\infty=\frac{d}{2D}.
\end{equation}
\end{corollary}

\begin{corollary}[Attained redundancy]\label{cor:attained}
At the optimal weight,
\begin{equation}\label{eq:attained}
R_n(\xi^*) \;=\; \frac{d}{2}\ln\frac{n}{\ell^*_\eff}
 \;+\; o(1) + \epsilon,
\end{equation}
the redundancy of a code trained on a matched sample of length~$\ell^*_\eff$.
\end{corollary}

\begin{proof}
Substitute \eqref{eq:xistar} into \eqref{eq:main}: the linear term
equals~$\frac{d}{2}$ and cancels the constant.
\end{proof}

Four consequences matter for design.

\emph{Endpoints.} At~$D=0$, $\xi^*=1$ and~\eqref{eq:attained} is the
matched-sample $\frac{d}{2}\ln\frac{n}{\ell}$: matched data are used
in full. As~$\ell\to\infty$ at fixed~$D>0$, $\xi^*\to 0$ while
$\ell^*_\eff\to\ell_\infty$: the correct response to a growing
mismatched corpus is a vanishing weight, not a vanishing use of the
corpus.

\emph{Flatness.} Since \eqref{eq:xistar} gives $\xi^*\propto 1/D$
once~$\ell D\gg d$, a factor-$k$ error in an estimate of~$D$ changes
the redundancy~\eqref{eq:main} by $\frac{d}{2}(1/k-1+\ln k)$ nats, a
constant independent of~$\ell$, $D$, and~$n$. Estimation of~$D$
therefore needs only order-of-magnitude accuracy.

\emph{Price of full weight.} Setting~$\xi=1$ under mismatch costs
$\ell D-\frac{d}{2}\ln(1+2\ell D/d)$ nats over \eqref{eq:attained},
growing linearly in~$\ell D$, while setting~$\xi=0$ costs only
$\frac{d}{2}\ln\ell^*_\eff$ nats over it. The asymmetry is the
practical lesson: a long mismatched training sequence at full weight
is far worse than no training data at all.

\emph{Weighting beats truncation.} Truncation is the special case~$\xi=1$ applied to a prefix: by~\eqref{eq:main} with~$\ell$ replaced
by~$\ell_\eff$, it costs
$\frac{d}{2}\ln\frac{n}{\ell_\eff}+\ell_\eff D$, minimized at
$\ell_\eff=\ell_\infty$, recovering the truncation rule of~\cite{ReznikAnisimov2003a}, with minimum
$\frac{d}{2}\ln\frac{n}{\ell_\infty}+\frac{d}{2}$ nats. This exceeds
\eqref{eq:attained} by an amount growing to~$d/2$ nats as~$\ell$
grows. Discounting a heterogeneous database beats curating a matched
subset of the same effective size.

One more fact is needed to justify the algorithms: nothing forces the
weight to be constant during encoding, and one may ask whether a
decaying schedule~$\xi_t$ beats the constant~$\xi^*$. No schedule
helps.

\begin{theorem}[Per-symbol optimality]\label{thm:persymbol}
Let~$\rho_t(\xi)$ be the per-symbol excess \eqref{eq:rhodef} under
weight~$\xi$, and $\gamma_t=\ell_\eff/(\ell_\eff+t-1)$. Then
\begin{equation}\label{eq:rhot}
\rho_t(\xi) = \gamma_t^2\,D_2
 + \frac{d}{2}\cdot\frac{\xi^2\ell+t-1}{(\ell_\eff+t-1)^2}
 + O\bigl(\gamma_t^3\bigr)
 + O\Bigl(\tfrac{\|P_T-P_S\|}{\ell_\eff+t-1}\Bigr)
 + O\Bigl(\tfrac{1}{(\ell_\eff+t-1)^{3/2}}\Bigr),
\end{equation}
and for every~$t\ge 1$ the minimizer of the leading terms of~\eqref{eq:rhot} is the same value, determined by
$1/(\xi^*\ell)=1/\ell+2D_2/d$, independent of~$t$. By
Lemma~\ref{lem:JS} it coincides with~\eqref{eq:harmonic} in the
small-mismatch regime.
\end{theorem}

\begin{proof}
Expression \eqref{eq:rhot} is~\eqref{eq:perstep} with the mismatch
term expanded by~\eqref{eq:JSquad}. The~$\xi$-derivative of the
leading terms at fixed~$s=t-1$ is
$\frac{\ell s}{(\xi\ell+s)^3}[2\xi\ell D_2+d\,\xi-d]$,
whose bracket vanishes at $\xi=d/(2\ell D_2+d)$ independently of~$s$, and changes sign from negative to positive there.
\end{proof}

The anticipated benefit of a decaying schedule is already delivered by
the count updates: the influence of the training counts on
\eqref{eq:kt} is the factor $\ell^*_\eff/(\ell^*_\eff+t-1)$, which
decays on its own, and~\eqref{eq:rhot} shows this automatic decay to
be the right one. A manual decay would double-count the fade. The
constant~$\xi^*$ is therefore the target that any selection scheme
should aim at.

\section{Implementations}
\label{sec:algs}

The optimal weight~\eqref{eq:xistar} takes two inputs. The training length~$\ell$
is known. The mismatch~$D$ is not: it is a property of the pair of
sources, and neither source is observed. What the coder holds is the
training data and, as coding proceeds, a growing prefix of the message.
This section gives three ways to supply the missing number. Two facts do most of the work: the flatness
of the optimum, quantified above, and free decoder synchronization.
The decoder holds the training data as common side information and
reconstructs the message symbol by symbol, so any weight computed from
$(x_T,\,x_{S,1}\cdots x_{S,t-1})$ is reproducible there and costs no
side information.

The shared building block is a divergence estimate. For two independent
samples of sizes~$N$ and~$M$ with KT-smoothed empirical distributions
$\tilde P(x)=(k_x+\frac12)/(N+\frac{m}{2})$ and
$\tilde Q(x)=(k'_x+\frac12)/(M+\frac{m}{2})$, where~$k_x$ and~$k'_x$
are the respective counts,
\begin{equation}\label{eq:debias}
\hat D \;=\; \Bigl[\, D(\tilde P\,\|\,\tilde Q)
 \;-\;\frac{d}{2}\Bigl(\frac{1}{N}+\frac{1}{M}\Bigr)\Bigr]_+ .
\end{equation}
The plug-in divergence between independent empiricals overshoots, to
second order, by exactly the Wilks excess of the two fits, and~\eqref{eq:debias} removes it. The
smoothing keeps the estimate finite when one sample contains symbols
the other lacks.

\subsection{Offline: the Database Spread}
\label{sec:offline}

\begin{figure}[t]
\centering
\begin{tikzpicture}[
  font=\small,
  box/.style={draw, rounded corners=1pt, minimum height=8mm,
              align=center, inner sep=3pt},
  arr/.style={-{Stealth[length=2.2mm]}}]
\node[font=\small\itshape] at (2.0,4.55) {design time};
\node[font=\small\itshape] at (12.4,4.55) {coding time};
\draw[dashed] (6.2,-0.1) -- (6.2,4.8);
\node[box] (db) at (2.0,3.5) {training set\\$u_1,\dots,u_R$};
\node[box] (pair) at (2.0,2.05) {pairwise~$\hat D_{ij}$\\by~\eqref{eq:debias}};
\node[box] (xio) at (2.0,0.6) {$\xi^*=\frac{d}{2\ell\hat D+d}$\ \eqref{eq:xistar}};
\node[box, thick] (kt) at (9.4,3.5) {weighted KT\\\eqref{eq:kt}};
\node[box] (ac) at (12.4,3.5) {arithmetic\\coder};
\node[inner sep=0] (msg) at (9.4,4.55) {message $x_{S,1}\cdots x_{S,n}$};
\draw[arr] (db) -- (pair);
\draw[arr] (pair) -- (xio) node[midway, right, font=\scriptsize] {average~$\hat D$};
\draw[arr] (db.east) -- (kt.west)
   node[pos=0.42, below, font=\scriptsize] {counts~$k_x(x_T)$, length~$\ell$};
\draw[arr] (xio.east) -| (kt.south)
   node[pos=0.22, above, font=\scriptsize] {$\xi^*$};
\draw[arr] (msg) -- (kt);
\draw[arr] (kt) -- (ac);
\draw[arr] (ac.east) -- ++(1.0,0) node[right, inner sep=1pt] {bits};
\end{tikzpicture}
\caption{The offline selector: the spread of the training set fixes
the weight at design time.}
\label{fig:offline}
\end{figure}
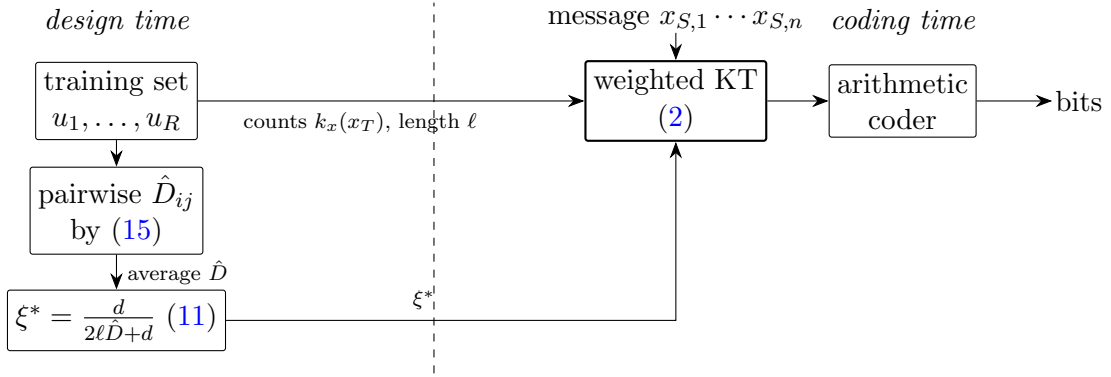

A single training sequence carries no information about~$D$: the
mismatch is a property of the pair of sources, and only one left a
trace. What breaks the impasse offline is structure that training
databases usually have: a collection of sequences $u_1,\dots,u_R$ of
one nominal type. Fig.~\ref{fig:offline} shows the scheme. The
pairwise divergences \eqref{eq:debias} between the sequences estimate
the spread of the population, a message from the same population sees
a mismatch of that order, and the average sets the weight, once and
at design time. Coding then runs \eqref{eq:kt} with the concatenated
collection as~$x_T$. If the estimate clamps to zero,
\eqref{eq:xistar} returns~$\xi^*=1$, the correct action for an
indistinguishable population. When past messages are available, the
estimate is direct: \eqref{eq:debias} against held-out messages.

\subsection{Plug-in: a Feedback Loop}
\label{sec:plugin}

\begin{figure}[t]
\centering
\begin{tikzpicture}[
  font=\small,
  box/.style={draw, rounded corners=1pt, minimum height=8mm,
              align=center, inner sep=3pt},
  arr/.style={-{Stealth[length=2.2mm]}},
  farr/.style={-{Stealth[length=2.2mm]}, dashed}]
\node[box] (tr) at (1.6,3.6) {training\\$x_T$};
\node[box, thick] (kt) at (5.7,3.6) {weighted KT \eqref{eq:kt}\\weight~$\xi_t$};
\node[box] (ac) at (9.7,3.6) {arithmetic\\coder};
\node[inner sep=0] (msg) at (5.7,4.85) {message $x_{S,1}\cdots x_{S,n}$};
\draw[arr] (msg) -- (kt);
\draw[arr] (tr) -- (kt);
\draw[arr] (kt) -- (ac);
\draw[arr] (ac.east) -- ++(1.1,0) node[right, inner sep=1pt] {bits};
\node[box] (pre) at (9.7,1.6) {prefix empirical\\$x_{S,1}\cdots x_{S,t}$};
\node[box] (dh) at (5.7,1.6) {$\hat D_t$ by~\eqref{eq:debias},\\$(N,M)=(\ell,t)$};
\node[box] (xi) at (1.6,1.6) {$\xi_t=\frac{d}{2\ell\hat D_t+d}$};
\draw[farr] (ac) -- (pre) node[midway, right, font=\scriptsize] {decoded};
\draw[farr] (pre) -- (dh);
\draw[farr] (dh) -- (xi);
\draw[farr] (xi.north) -- ++(0,0.55) -| ([xshift=-12mm]kt.south);
\draw[-{Stealth[length=1.8mm]}] (1.2,0.15) -- (11.9,0.15);
\node[font=\scriptsize, below] at (11.6,0.1) {$t$};
\foreach \x/\t in {1.7/1, 2.0/2, 2.6/4, 3.8/8, 6.2/16, 11.0/32}
  {\draw (\x,0.05) -- (\x,0.25);
   \node[font=\scriptsize, below] at (\x,0.05) {\t};}
\node[font=\scriptsize] at (8.6,0.5) {update times (doubling)};
\end{tikzpicture}
\caption{The plug-in selector: the training-to-prefix divergence is
fed back into the weight at doubling times (dashed loop).}
\label{fig:plugin}
\end{figure}
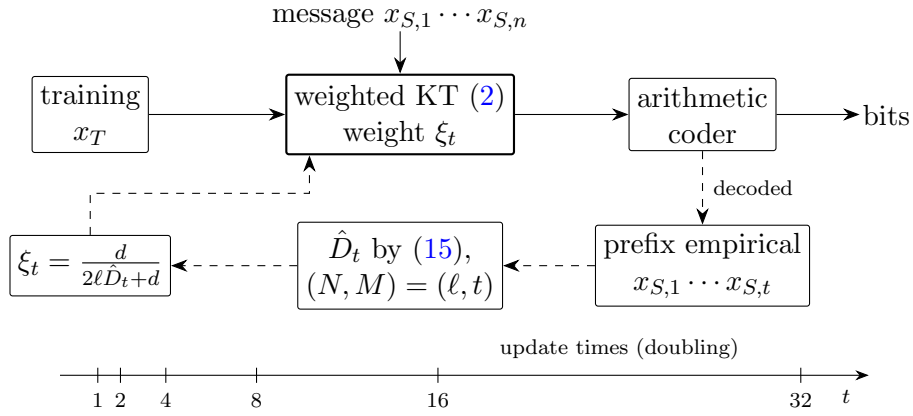

When no population structure is available, the message itself is the
only witness to the mismatch, and the coder estimates~$D$ on the fly.
Fig.~\ref{fig:plugin} shows the loop: at prefix lengths
$t=1,2,4,8,\dots$, recompute~$\hat D_t$ by~\eqref{eq:debias} between
the training and prefix empiricals and reset the weight by~\eqref{eq:xistar}, starting from~$\xi=1$. Logarithmically spaced
updates are adequate, matching the rate at which the estimate itself
changes: \eqref{eq:debias} improves by a constant factor only per
doubling of~$t$, so finer updates would recompute the same number
within its own noise. By flatness, the weight held between updates,
at most one doubling stale, costs a bounded number of nats per
interval. The decoder reproduces the schedule, and
Theorem~\ref{thm:persymbol} supplies the target:
a constant, so there is no schedule left to tune.

The apparent circularity, that the weight matters early while the
estimate becomes accurate late, resolves because the two timescales
coincide: the sensitivity of the redundancy to the weight is
concentrated on the first~$O(\ell^*_\eff)$ message symbols, and
distinguishing two sources at divergence~$D$ takes on the order of
$\ell_\infty=d/(2D)$ symbols, the scale of Stein's lemma. The trust
cap of~\eqref{eq:harmonic} is also the \emph{detection horizon}, so
the price of initializing at full weight is bounded, of the order of~$d$ nats total uniformly in~$\ell$ and~$D$, and visible only on
messages shorter than a few multiples of~$\ell_\infty$.

\subsection{Mixture: Twice-Universal over the Weight}
\label{sec:mixture}

\begin{figure}[t]
\centering
\begin{tikzpicture}[
  font=\small,
  box/.style={draw, rounded corners=1pt, minimum height=7mm,
              align=center, inner sep=3pt},
  arr/.style={-{Stealth[length=2.2mm]}},
  farr/.style={-{Stealth[length=2.2mm]}, dashed}]
\node[box] (tr) at (1.7,2.85) {training\\$x_T$};
\node[inner sep=0] (msg) at (3.7,5.4) {message $x_{S,1}\cdots x_{S,n}$};
\node[box] (a1) at (6.5,4.6) {KT, $\xi=1$};
\node[box] (a2) at (6.5,3.7) {KT, $\xi=\tfrac12$};
\node[box] (a3) at (6.5,2.8) {KT, $\xi=\tfrac14$};
\node       (ad) at (6.5,2.05) {$\vdots$};
\node[box] (a0) at (6.5,1.25) {KT, $\xi=0$};
\node[draw, circle, inner sep=2pt] (sum) at (9.7,2.9) {$\Sigma$};
\node[box] (ac) at (12.5,2.9) {arithmetic\\coder};
\foreach \a in {a1,a2,a3}
  {\draw[arr] (tr.east) -- ([yshift=-1.8mm]\a.west);}
\node[font=\scriptsize, below] at (3.4,2.28) {training counts~$k_x(x_T)$};
\foreach \a in {a1,a2,a3,a0}
  {\draw[arr] (msg.south) -- ([yshift=1.8mm]\a.west);}
\foreach \a in {a1,a2,a3,a0} {\draw[arr] (\a.east) -- (sum);}
\draw[arr] (sum) -- (ac);
\draw[arr] (ac.east) -- ++(1.0,0) node[right, inner sep=1pt] {bits};
\node[font=\scriptsize, align=center] at (9.7,1.0)
  {posterior weights~$w_j$,\\updated from decoded symbols};
\end{tikzpicture}
\caption{The mixture selector: a bank of weighted KT codes over a
grid of weights, combined by posterior weighting.}
\label{fig:mixture}
\end{figure}
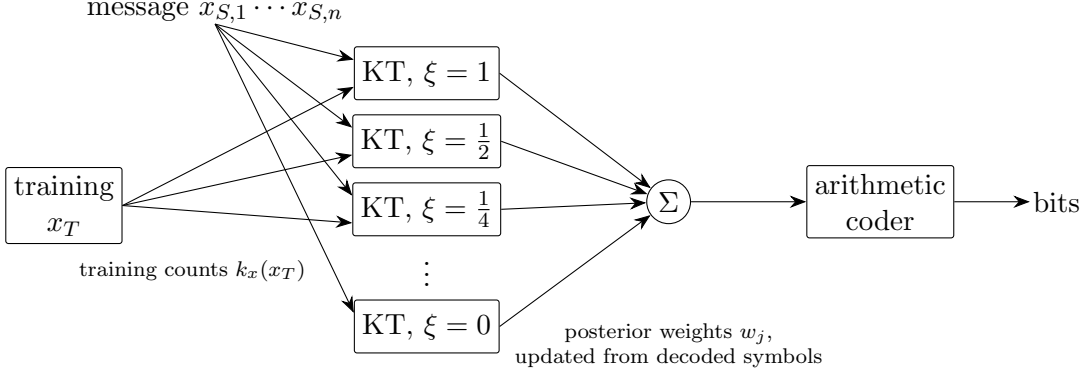

Estimation can be avoided altogether. Fix the geometric grid~$\xi_j=2^{-j}$, $j=0,\dots,K$, $K=\lceil\log_2\ell\rceil$, plus the
cold-start arm~$\xi=0$, and encode with the uniform mixture
$Q_{\mathrm{mix}}=\frac{1}{K+2}\sum_j Q_{\xi_j}$, computed as in
Fig.~\ref{fig:mixture}: the next-symbol probability is the
posterior-weighted average of the per-arm conditionals \eqref{eq:kt},
and the posterior updates from decoded symbols alone, so no arm is
ever selected or signaled. Pointwise,
$-\ln Q_{\mathrm{mix}} \le -\ln Q_{\xi_j}+\ln(K{+}2)$ for every
message and arm, so
\begin{equation}\label{eq:mixbound}
R^{\mathrm{mix}}_n \;\le\; R_n(\xi^*) \;+\; c_m \;+\; \ln(K{+}2),
\end{equation}
uniformly in~$D$, where~$c_m$ is the flatness cost of the nearest
grid arm, the factor-$\sqrt2$ instance of the flatness constant of
Section~\ref{sec:law}, and the premium~$\ln(K{+}2)$, doubly
logarithmic in~$\ell$, is paid once per message. The code is
\emph{twice universal} in the sense of Ryabko~\cite{Ryabko1984}:
within each arm, over the source; across arms, over the weight. The
construction parallels context-tree weighting~\cite{WillemsShtarkovTjalkens1995}, a sequential mixture of KT
estimators over model structures. Here the mixture runs over evidence
weights instead, and the two compose, since \eqref{eq:kt} is a
drop-in replacement for the KT estimator at each node of a context
tree. The boundary arms are the exact codes of the limiting regimes:
$\xi=0$ optimal when borrowing even one symbol is too much, $j=0$
optimal at~$D=0$.

\subsection{Complexity of the Proposed Approaches}
\label{sec:complexity}

The baseline is the coder itself. The weight enters \eqref{eq:kt}
only through the initial pseudo-counts, $\xi k_x(x_T)+\frac12$ in
place of~$\frac12$, so after this one-time initialization the
weighted code runs exactly as the classical KT coder: $O(m)$ work per
symbol over an~$m$-entry count table, hence~$O(mn)$ per message. Against this baseline the three selectors separate cleanly.
The offline weight adds nothing at coding time: it is one number,
computed from~$\binom{R}{2}$ evaluations of the~$O(m)$ estimator
\eqref{eq:debias} before any message is seen. The plug-in adds~$O(\log n)$ evaluations of~\eqref{eq:debias} during coding, an~$O(m\log n)$ overhead that vanishes against the~$O(mn)$ baseline.
The mixture is the only selector that multiplies: it runs~$O(\log\ell)$ arms per symbol, so coding costs~$O(\log\ell)$ times
the baseline. The guarantee~\eqref{eq:mixbound} is what the factor
buys.

\section{Experimental Results}
\label{sec:exp}

The three selectors are evaluated on the ten English-text files of the
Calgary and Canterbury corpora~\cite{BellClearyWitten1990,ArnoldBell1997}: \texttt{bib},
\texttt{book1}, \texttt{book2}, \texttt{news}, \texttt{paper1},
\texttt{paper2} from the former, and \texttt{alice29},
\texttt{asyoulik}, \texttt{lcet10}, \texttt{plrabn12} from the
latter. The files span distinct genres: a bibliography, two novels, news, two
technical papers, children's fiction, a play, proceedings, and verse.
The model is the memoryless source over the byte alphabet, $m=256$,
$d=255$, a deliberate mis-specification: real text is not
i.i.d., but all coders share the same order-$0$ model, so the
comparison isolates the weight-selection rules. The protocol is
leave-one-out: each file in turn is the message, the concatenation of
the remaining nine ($\ell\approx 2.4$--$3.1$ million symbols) the
training sequence. Six conditions are compared:
\begin{itemize}\setlength{\itemsep}{1pt}
\item \textbf{$\xi=0$}: the classical KT code, the cold start;
training data discarded. The sensible default of universal coding,
and the baseline.
\item \textbf{$\xi=1$}: the sample-based code; training data pooled
with the message at full weight, as if matched.
\item \textbf{$\xi^*$ offline}: the weight~\eqref{eq:xistar} from
the average pairwise divergence between the nine training files
(Section~\ref{sec:offline}); one number, fixed before coding.
\item \textbf{plug-in}: the loop of Section~\ref{sec:plugin},
initialized at~$\xi=1$.
\item \textbf{mixture}: Section~\ref{sec:mixture}, $K+2=24$ arms.
\item \textbf{oracle}: the best \emph{fixed} weight for the given
file, found in hindsight by encoding the file at every weight on the
grid~$2^{-j/4}$, $j=0,\dots,88$, and keeping the minimum. The oracle
knows the answer it is asked to find: it sees the realized code
lengths of the very file being scored, so it is not an implementable
coder. It is the reference that any fixed-weight scheme can at best
attain, and the gap between a selector and the oracle is the price of
not knowing~$D$.
\end{itemize}
This section reports in bits: all code lengths are ideal
arithmetic-coding lengths, $-\log_2 Q$, in bits per symbol, and the
divergence estimates, computed in nats and consumed by~\eqref{eq:xistar} in nats, are tabulated in bits per symbol
($1$ nat~$=\log_2 e$ bits). The benchmark is reproduced by the
open-source package~\cite{WktCode}.

Table~\ref{tab:weights} reports the per-file empirical order-$0$
entropies~$\hat H$, the offline divergence estimates~$\hat D$, and
the selected weights; for the plug-in, the weight reported is its
final value, set at the last doubling update.
Table~\ref{tab:codelengths} reports the per-file redundancies, ideal
code length minus~$\hat H$, with their length-weighted averages and,
in the last row, the relative change of the average against the cold
start, in percent.

\begin{table}[t]
\centering
\caption{Empirical entropy, divergence, and selected weights ($\times 10^{-3}$).}
\label{tab:weights}
\footnotesize
\begin{tabular}{lrrrrrr}
\toprule
file & $n$ & $\hat H$ & $\hat D$ & $\xi^*$ offline & $\xi$ plug-in & $\xi$ oracle \\
\midrule
\texttt{bib}      & 111\,261 & 5.201 & 0.284 & 0.21 & 0.14 & 0.12 \\
\texttt{book1}    & 768\,771 & 4.527 & 0.414 & 0.19 & 0.34 & 0.29 \\
\texttt{book2}    & 610\,856 & 4.793 & 0.425 & 0.17 & 0.74 & 0.69 \\
\texttt{news}     & 377\,109 & 5.190 & 0.400 & 0.16 & 0.52 & 0.49 \\
\texttt{paper1}   &  53\,161 & 4.983 & 0.387 & 0.15 & 0.68 & 0.49 \\
\texttt{paper2}   &  82\,199 & 4.601 & 0.434 & 0.14 & 0.52 & 0.82 \\
\texttt{alice29}  & 148\,481 & 4.513 & 0.400 & 0.15 & 0.26 & 0.24 \\
\texttt{asyoulik} & 125\,179 & 4.808 & 0.379 & 0.16 & 0.23 & 0.21 \\
\texttt{lcet10}   & 419\,235 & 4.623 & 0.412 & 0.16 & 0.36 & 0.41 \\
\texttt{plrabn12} & 471\,162 & 4.477 & 0.405 & 0.17 & 0.26 & 0.21 \\
\bottomrule
\end{tabular}
\end{table}

\begin{table}[t]
\centering
\caption{Redundancies, bits per symbol.}
\label{tab:codelengths}
\footnotesize
\setlength{\tabcolsep}{2.6pt}
\begin{tabular}{lrrrrrr}
\toprule
file & $\xi=0$ & $\xi=1$ & $\xi^*$ offline & $\xi$ plug-in & $\xi$ mixture & $\xi$ oracle \\
\midrule
\texttt{bib}      & 0.0125 & 0.6030 & 0.0109 & 0.0118 & 0.0106 & 0.0106 \\
\texttt{book1}    & 0.0023 & 0.0687 & 0.0019 & 0.0019 & 0.0019 & 0.0019 \\
\texttt{book2}    & 0.0028 & 0.0683 & 0.0022 & 0.0020 & 0.0021 & 0.0020 \\
\texttt{news}     & 0.0043 & 0.1644 & 0.0034 & 0.0035 & 0.0032 & 0.0032 \\
\texttt{paper1}   & 0.0234 & 0.1578 & 0.0168 & 0.0173 & 0.0155 & 0.0154 \\
\texttt{paper2}   & 0.0162 & 0.0541 & 0.0119 & 0.0117 & 0.0103 & 0.0102 \\
\texttt{alice29}  & 0.0097 & 0.1055 & 0.0078 & 0.0084 & 0.0077 & 0.0077 \\
\texttt{asyoulik} & 0.0113 & 0.1779 & 0.0092 & 0.0105 & 0.0093 & 0.0092 \\
\texttt{lcet10}   & 0.0039 & 0.0815 & 0.0031 & 0.0033 & 0.0030 & 0.0030 \\
\texttt{plrabn12} & 0.0035 & 0.0884 & 0.0030 & 0.0029 & 0.0030 & 0.0030 \\
\midrule
average           & 0.00479 & 0.11054 & 0.00383 & 0.00392 & 0.00370 & 0.00368 \\
$\Delta$ vs $\xi{=}0$, \% & $0$ & $+2210$ & $-19.9$ & $-18.1$ & $-22.7$ & $-23.2$ \\
\bottomrule
\end{tabular}
\end{table}

The core result is in the last two rows of
Table~\ref{tab:codelengths}. Relative to the better classical choice
on this corpus, the cold start~$\xi=0$, the mixture removes~$22.7\%$
of the redundancy, within~$0.0003$ bits per symbol of the
oracle's~$23.2\%$ ceiling. The offline weight removes~$19.9\%$, the
plug-in~$18.1\%$. Full weight is not a contender: it
multiplies the cold start's redundancy twenty-three-fold, the linear
penalty~$\ell\hat D\sim 10^6$ bits dominating. The weighted code
beats both classical choices on every file, so neither endpoint is
ever the right answer here.

Behind the averages, the tables confirm the analysis. The divergence
estimates fall in~$0.28$ to~$0.43$ bits per symbol, five to nine
percent of the files' entropies, and the optimal effective lengths sit
at~$420$ to~$650$ characters against training sequences of three
million: the harmonic cap in action, with~$99.98\%$ of the training
data discounted away and the code still better than the cold start.
\texttt{book2} shows the division of labor: its realized mismatch
($0.69$ in the oracle column of Table~\ref{tab:weights}) is four
times the population average ($0.17$), and the plug-in tracks it
to~$0.74$ where the offline weight cannot: the loop measures the
message, the spread measures the population.

\section{Conclusions}
\label{sec:concl}

Training data in adaptive coding generally do not match the source
being encoded. They have traditionally been controlled by selection:
which data, and how much. This paper made the \emph{weight} at which
they enter the code the control variable: a KT estimator whose
training counts are scaled by~$\xi\in[0,1]$. The choice admits a
closed-form answer. For a training sequence of length~$\ell$ at
divergence~$D$ nats per symbol from the message source, the optimal
weight is $\xi^*=d/(2\ell D+d)$, $d=m-1$. The effective training
length $\xi^*\ell$ obeys the harmonic law~\eqref{eq:harmonic}. Any mismatch caps the usable training
information at~$d/(2D)$ symbols, however much was collected. The law
converts a tuning problem into a measurement problem, and three
selectors solve it: offline, from the spread of the training set; by
a plug-in loop, from the decoded prefix; and by a twice-universal
mixture, with no estimation at all. On standard corpora the resulting
codes beat both classical endpoints, $\xi=0$ and~$\xi=1$, on every
file, removing up to~$22.7\%$ of the redundancy of the classical KT
code, the $\xi=0$ endpoint.

The results extend directly to the conditional coders of practice. A
context-based coder
over~$K$ contexts runs one such estimator per context, the
model and the coding distribution factorize across contexts, and the
redundancy is the sum of~$K$ single-context redundancies, with
training length~$\ell_s$, message length~$n_s$, and divergence~$D_s$
at context~$s$. The
aggregate reproduces \eqref{eq:main} with $\ell D=\sum_s \ell_s D_s$,
plus one new $\xi$-free term,
$\frac{m-1}{2}\sum_s\ln\frac{n_s/n}{\ell_s/\ell}$, pricing the
difference between the training and message context frequencies.
Since the new term does not involve~$\xi$, the optimal
weight~\eqref{eq:xistar} and the harmonic law~\eqref{eq:harmonic}
apply per context and in aggregate unchanged.

This covers the primed coders of the introduction. In video coders of
the CABAC family~\cite{Marpe2003}, per-context weights, with~$D_s$
estimated at design time from the corpus partitioned by clip, can
replace a few hand-tuned initialization strengths by a computed
per-context value. In segmented and random-access formats~\cite{Zstd}, where the
cold-start overhead $\frac{d}{2}\ln n$ dominates, the weight can
convert a shared dictionary from a binary decision into a dial. With
a sliding window over a stream's own past in the role of the training
sequence, the same construction can cover sources that drift. And for
a pretrained model in the role of the training
source~\cite{Deletang2024}, the weight can price what the model is
worth on the data actually encoded.

\end{document}